\documentclass[letterpaper, 10 pt, conference]{ieeeconf}  
\usepackage{amsmath} 
\usepackage{amssymb}  
\usepackage{graphicx}
\usepackage{amsmath}
\usepackage{epstopdf}
\usepackage{mathtools}
\usepackage{dsfont}
\usepackage{url}
\usepackage[ruled,norelsize]{algorithm2e}
\usepackage{color}
\usepackage{booktabs} 
\usepackage{multirow}   
\usepackage{balance}    
\usepackage{soul}   
\usepackage{hyperref}   
\usepackage{empheq}

\usepackage[noadjust]{cite} 

\makeatletter
\newcommand{\removelatexerror}{\let\@latex@error\@gobble}
\makeatother

\newcommand{\beq}{\begin{equation}}
\newcommand{\eeq}{\end{equation}}

\newcommand{\argmaxF}{\mathop{\mathrm{argmax}}\limits}   
\usepackage{bm}

\newcounter{algorithmctr}[section]
\renewcommand{\thealgorithmctr}{\thesection.\arabic{algorithmctr}}
   {\refstepcounter{algorithmctr}\begin{list}{}{%
       \setlength{\rightmargin}{0\linewidth}%
       \setlength{\leftmargin}{.05\linewidth}
        \setlength{\itemsep}{1pt}
  \setlength{\parskip}{0pt}
  \setlength{\parsep}{0pt}}%
       \rmfamily\small
       \item[]{\setlength{\parskip}{0ex}\hrulefill\par%
        \nopagebreak{\bfseries\textsf{Algorithm \thealgorithmctr~}}}}%
   {{\setlength{\parskip}{-1ex}\nopagebreak\par\hrulefill} \end{list}}
\IEEEoverridecommandlockouts
\newtheorem{assumption}{Assumption}
\newtheorem{theorem}{Theorem}
\newtheorem{proposition}{Proposition}

\newtheorem{remark}{Remark}

\newcommand{\SUM}{ \sum_{i=0}^j \sum_{k=0}^K}

\title{\LARGE \bf Sample-Based Learning Model Predictive Control \\ for Linear Uncertain Systems }

\author{Ugo Rosolia and Francesco Borrelli
\thanks{U.\ Rosolia and F.\ Borrelli are with the Department of Mechanical Engineering, University of California at Berkeley ,
        Berkeley, CA 94701, USA
        {\tt\small\{ugo.rosolia,\ fborrelli\}{@}berkeley.edu}}%
}

\begin{document}

\maketitle
\thispagestyle{empty}
\pagestyle{empty}

\begin{abstract}
We present a sample-based Learning Model Predictive Controller (LMPC) for constrained uncertain linear systems subject to bounded additive disturbances. 
The proposed controller builds on earlier work on LMPC for deterministic systems. 
First, we introduce the design of the safe set and value function used to guarantee safety and performance improvement. 
Afterwards, we show how these quantities can be approximated using noisy historical data. 
The effectiveness of the proposed approach is demonstrated through a numerical example. 
We show that the LMPC is able to safely explore the state space and to iteratively improve the worst-case closed-loop performance, while robustly satisfying state and input constraints.
\end{abstract}

\section{Introduction}
Exploiting historical data in order to iteratively improve the performance of Model Predictive Controllers (MPC) has been an active theme of research in the past few decades \cite{aswani2013provably,kocijan2004gaussian,koller2018learning,hewing2018cautious,chua2018deep,terzi2018learning,rosolia2019learning,lee1997model,ostafew2014learning,lee2000convergence,lee2000model}. The key idea is to use stored state-input pairs in order to compute at least one of the following three components used in the control design: $\emph{i})$ a \textit{model} which describes the evolution of the system, $\emph{ii})$ a \textit{safe set} of states (and an associated control policy $\pi(\cdot)$) from which the control task can be safely completed and $\emph{iii})$ a \textit{value function} which represents the cumulative closed-loop cost from a given point of the safe set when the policy $\pi(\cdot)$ is used. In this work, we present a strategy to build safe sets and the associated value functions by exploiting historical noisy closed-loop trajectories.


Policy evaluation strategies used to estimate value functions from historical data are studied in Approximate Dynamic Programming (ADP) and Reinforcement Learning (RL)~\cite{bertsekas1996neuro,bertsekas2019feature,recht2018tour}. For instance, direct strategies compute the estimate value function which best fits the closed-loop cost data over the stored states. On the other hand, in indirect strategies the estimate value function is computed by iteratively minimizing the temporal difference~\cite{sutton1988learning,bradtke1996linear}. A survey on policy evaluation strategies goes beyond the scope of this work, we refer the reader to~\cite{bertsekas2019feature,bertsekas1996neuro} for a comprehensive review on this topic.

The integration of MPC with system identification strategies has been extensively studied in the literature \cite{aswani2013provably,chua2018deep,koller2018learning,hewing2018cautious,kocijan2004gaussian,rosolia2019learning, terzi2018learning}. 
In \cite{chua2018deep} the authors identified the system's model using a deep neural network, which incorporates uncertainty using an ensemble of models. Another system identification strategy consists of fitting a Gaussian Process (GP) to experimental data \cite{koller2018learning,hewing2018cautious,kocijan2004gaussian}. GP provides a nominal model and confidence bounds, which may be used to tighten the constraint set over the planning horizon. This strategy allows to provide high-probability safety guarantees \cite{koller2018learning, hewing2018cautious}. The effectiveness of GP-based strategies on experimental platform has been shown in \cite{hewing2018cautious}, where a MPC is used to race a 1/43-scale vehicle. Regression strategies may also be used to identify the system's model \cite{rosolia2019learning, terzi2018learning}. For instance, the authors in \cite{terzi2018learning} used linear regression to identify both the nominal model and the model uncertainty used for robust MPC design. In \cite{rosolia2019learning}, we used local linear regression to identify the model used by the controller, which was able to drive a 1/10-scale race car at the limit of handling.

Data-based strategies to construct safe sets have been investigated in~\cite{wabersich2018linear,bacic2003general,brunner2013stabilizing,blanchini2003relatively,LMPClinear, LMPC}.
The authors in \cite{wabersich2018linear} proposed a linear model predictive safety certification framework, where safe sets are computed exploiting closed-loop data generated by a robust controller. In~\cite{bacic2003general,brunner2013stabilizing} the authors computed safe sets combining stored trajectories with polyhedron and ellipsoidal invariant sets. Another approach is proposed in~\cite{blanchini2003relatively} where the stored trajectories are mirrored to construct invariant sets.
In \cite{LMPClinear, LMPC} we showed that data from a deterministic system can be trivially used to compute safe sets. However, these strategies cannot be used to compute safe sets for uncertain system.


In this work we present a sample-based Learning Model Predictive Controller (LMPC) for linear systems subject to bounded additive uncertainty. 
We refer to a control task execution as ``iteration'' and we iteratively update the LMPC policy. 
At iteration $j-1$, we show how to construct a robust safe set and value function, which are used to synthesize the LMPC policy at next $j$th iteration. We show that the proposed strategy guarantees that: \emph{i}) state and input constraints are robustly satisfied, \emph{ii}) the closed-loop system converges asymptotically to a neighborhood of the origin, \emph{iii}) the worst-case performance of the $j$th LMPC policy is non-increasing with the iteration index, and \emph{iv)} the domain of the LMPC policy is not shrinking at each $j$th iteration. 
The proposed control strategy is computationally intensive. 
Therefore, we propose a practical algorithm that exploits simulations of the closed-loop system, which are associated with unknown sampled disturbance realizations.
These closed-loop simulations, referred to as ``roll-outs", are used to approximate the safe set and the value function used in the LMPC design. 


\section{Problem Definition}\label{sec:problemDef}
We consider the following linear time invariant system
\begin{equation}\label{eq:sys}
    x^j_{k+1} = A x^j_k + B u_k^j + w^j_k
\end{equation}
where at time $k$ of the $j$th iteration the disturbance $w_k^j \in \mathcal{W}$, the state $x_k \in \mathbb{R}^n$ and input $u_k^j \in \mathbb{R}^d$. Furthermore, the system is subject to the following convex polytopic state and input constraints, for all $k \geq 0$
\begin{equation*}\label{eq:stateInputConstr}
    x_k \in \mathcal{X} \text{ and } \pi^j(x^j_k) \in \mathcal{U}.
\end{equation*}

At each $j$th iteration, we define the worst-case iteration cost associated with the control policy $\pi^j(\cdot)$ as the solution to the Bellman equation
\begin{equation}\label{eq:BellmanRecursion}
    J^j_{\pi^j}(x_0^j) = \max_{w \in \mathcal{W}} [h(x_0^j, \pi^j(x_0^j)) + J^j_{\pi^j}(Ax_0^j + B\pi^j(x_0^j) + w)].
\end{equation}
The goal of the control design is to solve the following infinite time robust optimal control problem,
\begin{equation}\label{eq:controlProb}
\begin{aligned}
    J_{0 \rightarrow \infty}^{j,*}(x_S^j) = \min_{ \pi^j(\cdot)} & \quad J^j_{\pi^j}(x_0^j)\\
    & \quad x^j_{k+1} = A x^j_k + B  \pi^j(x^j_k) + w^j_k \\
    & \quad u_k^j =  \pi^j(x^j_k) \\
    & \quad x_{k}^j \in \mathcal{X},u_{k}^j \in \mathcal{U} \\
    & \quad x_0^j = x_S^j \\
    & \quad \forall w_k^j \in \mathcal{W}, k \in \{0, 1, \ldots\}.
\end{aligned}
\end{equation}

We present a strategy to iteratively design a feedback policy
\begin{equation}\label{eq:policy}
    \pi^j(\cdot) : \mathcal{F}^j\subseteq \mathcal{X} \rightarrow \mathcal{U}
\end{equation}
which is a feasible solution to Problem~\eqref{eq:controlProb} for $x_0^j \in \mathcal{F}^j$. In particular the proposed strategy guarantees: \emph{i}) \textit{convergence} of the closed-loop system \eqref{eq:sys} and \eqref{eq:policy} to a neighborhood of the origin $\mathcal{O}$, \emph{ii}) \textit{safety}, state and input constraints are robustly satisfied, \emph{iii}) \textit{performance improvement}, if the controller performs the same task repeatedly (i.e. $x_0^j=x_0^{j+1}$), then the worst-case iteration cost \eqref{eq:BellmanRecursion} is non-increasing (i.e. $J_{\pi^{j+1}}^{j+1}(x_0^{j+1}) \leq J_{\pi^{j}}^{j}(x_0^j)$), and \emph{iv)} \textit{exploration}, the domain of the policy \eqref{eq:policy} is not shrinking with the iteration index (i.e. $\mathcal{F}^i \subseteq \mathcal{F}^j, \forall j\geq i$). 

Throughout this paper we use the standard function classes $\mathcal{K}$, $\mathcal{K}_\infty$ and $\mathcal{KL}$ notation (see \cite{kellett2014compendium}) and we define the distance from a point $x\in\mathbb{R}^n$ to a set $\mathcal{O}\subseteq \mathbb{R}^n$ as
\begin{equation*}
    |x|_\mathcal{O} \overset{\Delta}{=} \inf_{d\in\mathcal{O}} ||x-d||_1.
\end{equation*}
Furthermore, we make the following assumptions.

\begin{assumption}\label{ass:O_inf}
The set $\mathcal{O} \subset \mathbb{R}^n$ is a robust positive invariant set for the autonomous system $x_{k+1} = (A + BK) x_k+w_k$, 
\begin{equation*}
    \forall x_k \in \mathcal{O} \rightarrow (A + BK) x_k + w_k \in \mathcal{O}, \forall w_k \in \mathcal{W}
\end{equation*}
and $\forall x_k \in \mathcal{O}$ we have that $Kx_k \in \mathcal{U}$.
\end{assumption}

\begin{assumption}\label{ass:cost}The continuous stage cost $h(\cdot,\cdot)$ is jointly convex in its arguments. Furthermore, we assume that $\forall x \in \mathbb{R}^n, \forall u \in \mathbb{R}^d$
\begin{equation*}
\begin{aligned}
\alpha_x^l(|x|_\mathcal{O}) \leq h(x,0) &\leq \alpha_x^u(|x|_\mathcal{O})\\
&\text{and } \alpha_u^l ( |u|_{K\mathcal{O}} ) \leq h(0,u) \leq \alpha_x^u(|u|_{K\mathcal{O}})
\end{aligned}
\end{equation*}
where $\alpha_{x}^u,\alpha_{x}^l, \alpha_{u}^u$ and $\alpha_{u}^l \in \mathcal{K}_\infty$.
\end{assumption}
Notice that the above assumptions imply that the optimal policy from~\eqref{eq:controlProb} robustly steers system~\eqref{eq:sys} to the goal set $\mathcal{O}$.

\section{Learning Model Predictive Control}
In this section we illustrate the control design strategy. We show how to construct a safe set of states, from which the control policy $\pi^j(\cdot)$ can successfully complete the control task. Afterward, we define a value function which approximates the cost-to-go associated with the control policy $\pi^j(\cdot)$. Finally, we exploit the safe set and the value function to synthesize the control policy $\pi^{j+1}(\cdot)$ at the next iteration $j+1$.

\subsection{Safe Set}
In this section we show how to iteratively construct a set of states from which the control task can be safely executed. First, we recall the definition of robust reachable set \cite{borrelli2017predictive} for the closed-loop system \eqref{eq:sys} and \eqref{eq:policy},
\begin{equation}\label{eq:reachSet}
\begin{aligned}
    \mathcal{R}_{k+1}&(x_0^j) = \Bigg\{x_{k+1} \in \mathcal{X} \Bigg| \begin{matrix} \exists w_k \in \mathcal{W}, x_k \in \mathcal{R}_{k}(x_0^j),\\
    {x_{k+1} = A x_k + B \pi^j(x_k) + w_k} \end{matrix} \Bigg\}
\end{aligned}
\end{equation}
with $\mathcal{R}_{0}(x_0^j) = x_0^j$. The above robust reachable set $\mathcal{R}_{N}(x_0^j)$ collects that states which may be reached in $N$-steps by the closed-loop system \eqref{eq:sys} and \eqref{eq:policy}.

Now, we define the safe set at the $j$th iteration as
\begin{equation}\label{eq:safeSet}
    \mathcal{SS}^j=\bigg\{ \bigcup_{k=0}^{\infty} \mathcal{R}_k(x_0^j) \bigg\} \bigcup \mathcal{O}.
\end{equation}
The above safe set $\mathcal{SS}^j$ contains the state evolution of the closed-loop system \eqref{eq:sys} and \eqref{eq:policy} at the $j$th iteration.
\begin{remark}
In practical applications each iteration has a finite-time duration. It is common in the literature to adopt an infinite time formulation at each iteration for the sake of simplicity. We follow such an approach in this paper. Our choice does not affect the practicality of the proposed method. In Section~\ref{sec:sampleBasedSafeSet}, we show that if the $j$th iteration is completed in finite time (i.e. $x_{T^j}^j \in \mathcal{O}, T^j < \infty$), then the safe set $\mathcal{SS}^j$ can be approximated using historical data.
\end{remark}

Finally, we define the convex safe set $\mathcal{CS}^j$ as the convex hull of the safe sets $\mathcal{SS}^k$ for iterations $k\in \{0,\ldots,j\}$, 
\begin{equation}\label{eq:CS}
    \mathcal{CS}^j = \text{conv}\Bigg( \bigcup_{k=0}^j \mathcal{SS}^k \Bigg).    
\end{equation}
Notice that, if the control policies $\pi^k(\cdot)$ for $k\in \{0, \ldots, j\}$ safely steer the system to the neighborhood of the origin $\mathcal{O}$. Then, $\mathcal{CS}^j$ is a robust control invariant set as stated by the following proposition.

\begin{proposition}\label{prop:SSinv}
For $j\geq0$, let $\pi^j(\cdot) : \mathcal{F}^j \rightarrow \mathcal{U}$ be a control policy defined over $\mathcal{F}^j \subseteq \mathcal{X}$. Consider system~\eqref{eq:sys} in closed-loop with $\pi^j(\cdot)$ and assume that $\forall x_0^j \in \mathcal{F}^j$ we have $x_k^j \in \mathcal{X}$ and $\lim_{t \rightarrow \infty }x_{t}^j \in \mathcal{O},\forall w_k\in \mathcal{W}, k\geq0$. Then, the convex safe set $\mathcal{CS}^j \subseteq \mathcal{X}$ is a robust control invariant set for system~\eqref{eq:sys},
\begin{equation*}
    \forall x \in \mathcal{CS}^j \rightarrow Ax+B\pi^j(x)+w \in  \mathcal{CS}^j, ~\forall w \in \mathcal{W}
\end{equation*}
\end{proposition}
\begin{proof}
By assumption $\pi^k(\cdot)$ for $k \in \{0,\ldots,j\} $ in closed-loop with~\eqref{eq:sys} robustly satisfies and input constraints. By definition~\eqref{eq:safeSet}, $\mathcal{SS}^k$ is a robust control invariant set for $k\in\{0,\ldots,j\}$. Therefore, by linearity of system \eqref{eq:sys}, $\mathcal{CS}^j \subseteq \mathcal{X}$ is a robust control invariant set.
\end{proof}

\subsection{Q-function}
In this section we define the value function $Q^j(\cdot) : \mathcal{CS}^j \rightarrow \mathbb{R}$, which approximates the cost-to-go from any state $x \in \mathcal{CS}^j$. Recall that the iteration cost~\eqref{eq:BellmanRecursion} for the control policy $\pi^j(\cdot)$ is given by the solution to following Bellman equation
\begin{equation}\label{eq:Bellman}
    J^j_{\pi^j}(x_0^j) = \max_{w \in \mathcal{W}} [h(x_0^j, \pi^j(x_0^j)) + J^j_{\pi^j}(Ax_0^j + B\pi^j(x_0^j) + w)],
\end{equation}
and it represents the worst-case cost-to-go from any point in the state space. The solution to the above Bellman equation is hard to compute~\cite{bertsekas1996neuro} and closed-form exists just for few problems~\cite{borrelli2017predictive}. For a survey on strategies to approximate the solution to Bellman equation we refer to~\cite{bertsekas1996neuro, bertsekas2019feature}.

Now, we define the worst-case cost-to-go over the safe set as 
\begin{equation}\label{eq:worstCostToGo}
    L_{\pi^j}^j(x) = \begin{cases} \max\limits_{w\in\mathcal{W}} [h(x , \pi^j(x)) +  L_{\pi^j}^j(x^j_+(w))] & \mbox{If } x\in \mathcal{SS}^j  \\ 
    +\infty & \mbox{If } x\notin \mathcal{SS}^j \end{cases}
\end{equation}
where $x^j_+(w) = Ax+B\pi^j(x)+w$. Notice that, for all $x \in \mathcal{SS}^j$, the above function coincides with the Bellman equation~\eqref{eq:Bellman}. 
The difference between $J_{\pi^j}^j(\cdot)$ and $L_{\pi^j}^j(\cdot)$ is that the domain of the latter is the safe set $\mathcal{SS}^j$ from~\eqref{eq:safeSet}. 
The solution equation~\eqref{eq:worstCostToGo} is still hard to compute, however it may be approximated using sampled closed-loop trajectories from~$\mathcal{SS}^j$, as shown in Section~\ref{sec:sampleQfun}. 

Finally, for all $x \in \mathcal{CS}^j$ we define the function
\begin{equation}\label{eq:Qfun}
    Q^j(x) = \min_\mu \{\mu ~ | ~ (x, \mu) \in  \text{conv}\big( \textstyle \bigcup_{k=0}^j \text{epi}(L_{\pi^j}(x)^j ) \big) \},
\end{equation}
which interpolates the worst-case cost-to-go functions $L_{\pi^k}^k(\cdot)$ for $k \in \{0, \ldots, j \}$. Notice that the above $Q^j(\cdot)$ is simply a convexification of the cost-to-go functions (i.e. $\text{epi}(Q^j(x)) = \text{conv}\big( \cup_{k=0}^j\text{epi}(L_{\pi^k}(x)^k ))$). Furthermore, if the control policies $\pi^k(\cdot)$ for $k\in \{0, \ldots, j\}$ safely steer the system to the neighborhood of the origin $\mathcal{O}$, then the approximated value function $Q^j(\cdot)$ is a robust control Lyapunov function over the convex safe set $\mathcal{CS}^j$ for system \eqref{eq:sys}, as shown by the following proposition.

\begin{proposition}\label{prop:lyapFun}
For $j\geq0$, let $\pi^j(\cdot) : \mathcal{F}^j \rightarrow \mathcal{U}$ be a control policy defined over $\mathcal{F}^j \subseteq \mathcal{X}$. Consider system~\eqref{eq:sys} in closed-loop with $\pi^j(\cdot)$ and assume that $\forall x_0^j \in \mathcal{F}^j$ we have $x_k^j \in \mathcal{X}$ and $\lim_{t \rightarrow \infty }x_{t}^j \in \mathcal{O} ~\forall w_k\in \mathcal{W}$. Then, $Q^j(\cdot)$ is a robust control Lyapunov function, i.e.
\begin{equation}\label{eq:proofContrLyap}
   \min_{u\in\mathcal{U}}\max_{w \in \mathcal{W}} \big[Q^j(Ax+Bu+w) + h(x,u) -Q^j(x) \big]\leq 0
\end{equation}
for all $x \in \mathcal{CS}^j$.
\end{proposition}
\begin{proof}
From definition~\eqref{eq:Qfun}, we have that $\forall x \in \mathcal{CS}^j$ there exist a set of multipliers $\{\lambda_0^0,\ldots,\lambda_k^i,\ldots,\lambda_K^j\}$ and a set of states $\{x_0^0,\ldots,x_k^i,\ldots,x_K^j\}$ such that for all $ k \in \{0,\ldots,K\}$ and for all $i \in \{0,\ldots,j\}$ we have $x_k^i \in \mathcal{SS}^i$, $\lambda_k^i \geq 0$, $\SUM \lambda_k^i = 1$, $\SUM \lambda_k^i x_k^i = x$, and
\begin{equation*}
Q^j(x) = \SUM \lambda_k^i L^i_{\pi^i}(x_k^i)\\.
\end{equation*}
Substituting in the above equation the definition of the worst-case cost-to-go~\eqref{eq:worstCostToGo} evaluated at $x_k^i \in \mathcal{SS}^i$ and leveraging the convexity of $h(\cdot, \cdot)$, we have that
\begin{equation*}
\begin{aligned}
Q^j(x) 
&= \SUM \lambda_k^i \big[ \max\limits_{w\in\mathcal{W}} [h(x_k^i , \pi^i(x_k^i)) +  L_{\pi^i}^i(x^i_{k,+}(w))] \big]\\
&\geq \max\limits_{w\in\mathcal{W}} [h( x ,  u) +  \SUM \lambda_k^i L_{\pi^i}^i(x^i_{k,+}(w))],
\end{aligned}
\end{equation*}
where $ x = \SUM \lambda_k^i x_k^i$, $ u = \SUM \lambda_k^i \pi^i(x_k^i) \in \mathcal{U}$ and $x^i_{k,+}(w) = Ax_k^i + B\pi^i(x_k^i) + w$. Definition~\eqref{eq:Qfun} implies that $Q^j(x) \leq L^i_{\pi^i}(x), \forall x \in \mathcal{CS}^j$ and $\forall i \in \{0,\ldots,j\}$, therefore from the above equation and convexity of $Q^j(\cdot)$ we conclude that
\begin{equation*}
\begin{aligned}
    Q^j(x)
    &\geq \max\limits_{w\in\mathcal{W}} [h( x ,  u) +  \SUM \lambda_k^i Q^j(x^i_{k,+}(w))]\\
    &\geq \max\limits_{w\in\mathcal{W}} [h( x ,  u) +   Q^j\big(\SUM \lambda_k^i x^i_{k,+}(w)\big)]\\
    &\geq \min_{u \in \mathcal{U}}\max\limits_{w\in\mathcal{W}} [h( x ,  u) + Q^j(A  x + B  u + w)].
\end{aligned}
\end{equation*}
\end{proof}

\subsection{Controller Design}
In this section we illustrate the controller design which leverages the convex safe set~\eqref{eq:CS} and the approximated value function~\eqref{eq:Qfun}. At each time $t$ of the $j$th iteration, we solve the following finite time optimal control problem
\begin{equation}\label{eq:FTOCP}
\begin{aligned}
    J^{\scalebox{0.5}{LMPC},j}_{t \rightarrow t+N }(x_t^j) = \min_{{\boldsymbol{\pi}_t^j(\cdot)}}   \max_{\bar {\bf{w}}_t^j} & ~ [\sum_{k=t}^{t+N-1} h(x_{k|t}^j, u_{k|t}^j) \\
    &\quad\quad\quad\quad\quad\quad+ Q^{j-1}(x_{t+N|t}^j)]\\
    & ~ x_{k+1|t}^j = A x_{k|t}^j + B u_{k|t}^j + \bar w_{k|t}^j \\
    & ~ u_{k|t}^j = \pi_{k|t}^j(x_{k|t}^j) \\
    & ~ x_{k|t}^j \in \mathcal{X}, u_{k|t}^j \in \mathcal{U} \\
    &~x_{t+N|t}^j \in \mathcal{CS}^{j-1} \\
    & ~ x_{t|t}^j = x_t^j \\
    & ~ \forall \bar w_{k|t}^j \in \mathcal{W}, k \in \{t, \ldots, t+N \}
\end{aligned}
\end{equation}
where the control policy ${\boldsymbol{\pi}_t^j}(\cdot)=[\pi_{t|t}^j(\cdot), \ldots, \pi_{t+N|t}^j(\cdot)]$ and the disturbance $\bar {\bf{w}}^j_t= [\bar w_{t|t}^j, \ldots, \bar w_{t+N|t}^j]$. The optimal feedback policy from the above finite time optimal control problem safely steers system~\eqref{eq:sys} from $x_t^j$ to the convex safe set, while minimizing the worst-case cost. Let 
\begin{equation}\label{eq:optPolicy}
    {\boldsymbol{\pi}}_t^{j,*}(\cdot)=[\pi_{t|t}^{j,*}(\cdot), \ldots, \pi_{t+N|t}^{j,*}(\cdot)]
\end{equation}
be the optimal feedback policy to Problem~\eqref{eq:FTOCP}. Then we apply to system \eqref{eq:sys}
\begin{equation}\label{eq:MPCpolicy}
    \pi^j(x_t^j) = \pi_{t|t}^{j,*}(x_t^j).
\end{equation}
The finite time optimal control problem \eqref{eq:FTOCP} is solved at time $t + 1$, based on the new state $x_{t+1|t+1}^j = x^j_{t+1}$, yielding a moving or receding horizon control strategy.

Furthermore, we define the domain of the LMPC policy~\eqref{eq:MPCpolicy}, which is given by
\begin{equation}\label{eq:F}
    \mathcal{F}^j = \Bigg\{x_0 \in \mathcal{X} \Bigg| \begin{matrix}\exists \kappa(\cdot):\mathbb{R}^n\rightarrow \mathbb{R}^d,x_k \in \mathcal{X}, \kappa(x_k) \in \mathcal{U} \\
    x_{k+1} = A x_k + B \kappa(x_k) + w_k, \\ 
    x_{N} \in \mathcal{CS}^{j-1}, \forall w_k \in \mathcal{W}, k\in\{0,\ldots,N\} \end{matrix} \Bigg\}.
\end{equation}
The set $\mathcal{F}^j$, which collects the feasible initial conditions to Problem~\eqref{eq:FTOCP}, is used to compute the initial state $x_0^j$ of the $j$th iteration. In particular, the initial condition at the $j$th iteration is computed solving the following convex optimization problem,
\begin{equation}\label{eq:initCondition}
    x_0^j = \argmaxF_{x \in \mathcal{F}^j} \{a x ~ | ~ a^\perp x = 0\}
\end{equation}
where the user-defined row vector $a \in \mathbb{R}^n$ represents the direction in which the LMPC explores the state space, and $a^\perp \in \mathbb{R}^n$ is a row vector perpendicular to $a$.

It is well-known that the solution to Problem~\eqref{eq:FTOCP} can be computed enumerating the vertices of the disturbance over the prediction horizon \cite{scokaert1998min}. Therefore, the computational complexity of  Problem~\eqref{eq:FTOCP} explodes with the horizon length $N$. For this reason, it is important to construct a terminal set and terminal cost, which allow to guarantee safety and performance improvement independently on the prediction horizon length. In the result section, we show that the proposed controller is able to safely explore the state space and to improve its performance, even with a short prediction horizon.

\subsection{Properties}
As discussed in Propositions~\ref{prop:SSinv}-\ref{prop:lyapFun}, for every point in $\mathcal{CS}^j$ there exists a control policy which safely steers the system to the terminal goal set. The properties of $\mathcal{CS}^j$ and $Q^j(\cdot)$ allow us to guarantee that the proposed strategy meets the requirements from Section~\ref{sec:problemDef}. 
The following theorem shows that the LMPC~\eqref{eq:FTOCP} and \eqref{eq:MPCpolicy} satisfies state and input constraints while steering the system to the neighborhood of the origin $\mathcal{O}$.

\begin{theorem}\label{th:recFeas}
Consider system~\eqref{eq:sys} in closed-loop with the LMPC~\eqref{eq:FTOCP} and \eqref{eq:MPCpolicy}. Let Assumptions~\ref{ass:O_inf}-\ref{ass:cost} hold, initialize $\mathcal{CS}^0=\mathcal{O}$ and $Q^0(\cdot)=0$. If $x_0^j\in\mathcal{F}^j, \forall j\geq1$, then the LMPC~\eqref{eq:FTOCP} and \eqref{eq:MPCpolicy} is feasible for all $t\geq0$ and iteration $j\geq1$. Furthermore, the closed-loop system asymptotically converges to $\mathcal{O}$, regardless of the disturbance realization.
\end{theorem}
\begin{proof}
Assume that at the $j$th iteration $Q^j(\cdot)$ is a robust control Lyapunov function defined on the robust control invariant set $\mathcal{CS}^j$. Then, by standard MPC arguments and the assumption on $x_0^j \in \mathcal{F}^j$, we have that at iteration $j+1$ the LMPC~\eqref{eq:FTOCP} and \eqref{eq:MPCpolicy} recursively satisfies state and input constraints, and the closed-loop system~\eqref{eq:sys} and \eqref{eq:MPCpolicy} converges asymptotically to the terminal set $\mathcal{O}$ \cite{borrelli2017predictive}. Consequently, the LMPC policy at iteration $j+1$ used to compute $Q^{j+1}(\cdot)$ and $\mathcal{CS}^{j+1}$ satisfies the assumptions in Propositions~\ref{prop:SSinv}-\ref{prop:lyapFun}, and therefore $Q^{j+1}(\cdot)$ is a robust control Lyapunov function defined on the robust control invariant set $\mathcal{CS}^{j+1}$. \\
The proof is completed by induction. We initialized $Q^0(\cdot)=0$, which is a robust control Lyapunov function defined on the robust control invariant set $\mathcal{CS}^0=\mathcal{O}$. Therefore it follows that $\forall j \geq 1$ the LMPC~\eqref{eq:FTOCP} and \eqref{eq:MPCpolicy} recursively satisfies state and input constraints, and the closed-loop system~\eqref{eq:sys} and \eqref{eq:MPCpolicy} converges asymotically to the terminal set $\mathcal{O}$.
\end{proof}

Next, we discuss the performance improvement properties. In particular, we show that if the initial condition of two subsequent iterations does not change (i.e. $x_0^j = x_0^{j+1}$), then the worst-case cost iteration cost is non-increasing.

\begin{theorem}\label{th:cost}
Consider system~\eqref{eq:sys} in closed-loop with the LMPC~\eqref{eq:FTOCP} and \eqref{eq:MPCpolicy}. Let Assumptions~\ref{ass:O_inf}-\ref{ass:cost} hold, initialize $\mathcal{CS}^0=\mathcal{O}$ and $Q^0(\cdot)=0$. If the initial condition of two subsequent iterations are equal, $x_0^{j+1} = x_0^j \in\mathcal{F}^j$. Then, the worst-case iteration cost~\eqref{eq:BellmanRecursion} is non-increasing with the iteration index $J_{0 \rightarrow T^{j+1}}^{j+1}(x_0^{j+1}) \leq J_{0 \rightarrow T^j}^{j}(x_0^j).$
\end{theorem}
\begin{proof}
By Theorem~\ref{th:recFeas}, the LMPC~\eqref{eq:FTOCP} and \eqref{eq:MPCpolicy} is feasible at time $t$ of the $j$th iteration. Let \eqref{eq:optPolicy} be the optimal policy time $t$ of the $j$th iteration, by Proposition~\ref{prop:lyapFun} we have 
\begin{equation*}
\begin{aligned}
    &J^{\scalebox{0.5}{LMPC},j}_{t \rightarrow t+N }(x_t^j)  =\sum_{k=t}^{t+N-1} h(x_{k|t}^{j,*}, \pi_{k|t}^{j,*}(x_{k|t}^{j,*})) + Q^{j-1}(x_{t+N|t}^{j,*}) \\
    &\quad \quad  \geq h(x_{t|t}^{j,*}, u_{t|t}^{j,*}) + \sum_{k=t+1}^{t+N-1} h(x_{k|t}^{j,*}, \pi_{k|t}^{j,*}(x_{k|t}^{j,*})))\\
    &\quad \quad + \min_{u\in\mathcal{U}}\max_{w \in \mathcal{W}} Q^{j-1}(Ax_{t+N|t}^{j,*}+Bu+w) + h(x_{t+N|t}^{j,*},u) \\
    &\quad \quad \geq  h(x_{t|t}^{j,*}, u_{t|t}^{j,*}) + \min_{{\boldsymbol{\pi}_t^j(\cdot)}}   \max_{{\bf{w}}_t^j} ~ [\sum_{k=t}^{t+N-1} h(x_{k|t}, u_{k|t}) \\
    &\quad\quad\quad\quad\quad\quad\quad\quad\quad\quad\quad\quad\quad\quad\quad\quad+ Q^{j-1}(x_{t+N|t})]\\
    &\quad \quad =  h(x_{t|t}^{j,*}, u_{t|t}^{j,*}) +J^{\scalebox{0.5}{LMPC},j}_{t+1 \rightarrow t+1+N }(x_{t+1}^j).
\end{aligned}
\end{equation*}
The above equation and the convergence of the closed-loop system~\eqref{eq:sys} and \eqref{eq:MPCpolicy} from Theorem~1 imply that
\begin{equation*}
\begin{aligned}
    J^{\scalebox{0.5}{LMPC},j}_{0 \rightarrow N }(x_0^j) &\geq h(x_{0|0}^{j,*}, x_{0|0}^{j,*}) +J^{\scalebox{0.5}{LMPC},j}_{1 \rightarrow 1+N }(x_{t+1}^j)\\
    & \geq \sum_{t=0}^\infty h(x_{t|t}^{j,*}, u_{t|t}^{j,*}) + \lim_{t \rightarrow \infty} J^{\scalebox{0.5}{LMPC},j}_{t \rightarrow t+N }(x_{t}^j)\\
    &=\sum_{t=0}^\infty h(x_{t}^{j}, u_{t}^{j}).
\end{aligned}
\end{equation*}
The above derivation holds for all disturbance realization, therefore we have that 
\begin{equation*}
    J^{\scalebox{0.5}{LMPC},j}_{0 \rightarrow N }(x_0^{j}) \geq J_{\pi^{j}}^{j}(x_0^{j}).
\end{equation*}

Finally we notice that the above inequality together with Equations~\eqref{eq:worstCostToGo}-\eqref{eq:Qfun} and the feasibility of the LMPC policy $\pi^j(\cdot)$ \eqref{eq:MPCpolicy} at the next iteration $j+1$ imply that
\begin{equation*}
\begin{aligned}
J_{\pi^{j}}^{j}(x_0^{j}) & = L^j_{\pi^j}(x_0^{j}) \\
&= \max_{w_0^j, \ldots, w_{N-1}^j} \sum_{k=0}^{N-1} [h(x_k^j,\pi^j(x_k^j)) + L^j_{\pi^j}(x_{N}^j)] \\
     & \geq \max_{w_0^j, \ldots, w_{N-1}^j} \sum_{k=0}^{N-1} [h(x_k^j,\pi^j(x_k^j)) + Q^j(x_{N}^j)] \\
    & \geq J^{\scalebox{0.5}{LMPC},j+1}_{0 \rightarrow N }(x_0^{j}) \geq J_{\pi^{j+1}}^{j+1}(x_0^{j})=J_{\pi^{j+1}}^{j+1}(x_0^{j+1}).
\end{aligned}
\end{equation*}
\end{proof}

Finally, we show that the domain of the LMPC~\eqref{eq:FTOCP} and \eqref{eq:MPCpolicy} does not shrink at each iteration.
\begin{theorem}\label{th:policyDomanin}
Consider system~\eqref{eq:sys} in closed-loop with the LMPC~\eqref{eq:FTOCP} and \eqref{eq:MPCpolicy}. Let Assumptions~\ref{ass:O_inf}-\ref{ass:cost} hold, and initialize $\mathcal{CS}^0=\mathcal{O}$ and $Q^0(\cdot)=0$. If $x_0^j\in\mathcal{F}^j, \forall j\geq1$. Then, the domain of which the LMPC defined in \eqref{eq:F} does not shrink at each iteration, i.e. $\mathcal{F}^i \subseteq \mathcal{F}^j, \forall j \geq i$.
\end{theorem}
\begin{proof}
The proof follows from the definition of the convex safe set. Notice that by definition~\eqref{eq:CS} we have that $\mathcal{CS}^i \subseteq \mathcal{CS}^j, \forall j \geq i$. Therefore, the terminal set in \eqref{eq:F} is not shrinking at each iteration and $\mathcal{F}^i \subseteq \mathcal{F}^j, \forall j \geq i$.
\end{proof}

\section{Practical Implementation}\label{sec:sampleBased}
In this section we show how the closed-loop trajectories associated with unknown sampled disturbance sequences can be used to approximate the convex safe set $\mathcal{CS}^j$ and the value function $Q^j(\cdot)$. At each $j$th iteration we collect $R$ simulations of the closed-loop systems, also referred to as ``roll-outs". Afterwards, we exploit these $R$ roll-outs to approximate the robust reachable sets~\eqref{eq:reachSet} and the worst-case cost-to-go~\eqref{eq:worstCostToGo}.

\subsection{Sample-Based Convex Safe Set}\label{sec:sampleBasedSafeSet}
In this section we show how the data from the closed-loop system~\eqref{eq:sys} and \eqref{eq:policy} can be used to approximate the convex safe set $\mathcal{CS}^j$. We define the $i$th disturbance realization sequence ${\bf{w}}^{j}_i=[w_{0,i}^{j},\ldots,w_{T^j,i}^{j}]$, where $w_{k,i}^{j}$ is the realized disturbance at time $k$ of the $j$th iteration. Furthermore, we denote the stored closed-loop trajectory associated with the $i$th disturbance realization  ${\bf{w}}^{j}_i$ as
\begin{equation}\label{eq:realizedTraj}
\begin{aligned}
    {\bf{x}}^{j}({\bf{w}}^{j}_i) &= [x_0^{j}({\bf{w}}^{j}_i), \ldots, x_{T^j}^{j}({\bf{w}}^{j}_i)], 
\end{aligned}
\end{equation}
where $T^j$ is the time at which the terminal goal set $\mathcal{O}$ is reached. The above notation emphasizes that the realized state $x_k^{j}({\bf{w}}^{j}_i)$ is a function of the realized disturbance sequence ${\bf{w}}^{j}_i$. Now, we notice that at each time $k$ of the $j$th iteration the state $x_k^{j}({\bf{w}}^{j}_i)$ is contained into the $k$-steps robust reachable set from $x_0^j$ (i.e. $x_k^{j}({\bf{w}}^{j}_i) \in \mathcal{R}_k(x_0^j))$. Therefore, we approximate the $k$-steps robust reachable set $\mathcal{R}_k(x_0^j)$ using $R$ roll-outs. In particular, for $i \in \{1, \ldots, R \}$ sampled disturbance sequences ${\bf{w}}^{j}_i$ we define the  approximated $k$-steps robust reachable set
\begin{equation}\label{eq:reachSetApprox}
    \tilde{\mathcal{R}}_k(x_0^j) = \text{conv} \Bigg( \bigcup_{i=1}^R x_k^{j}({\bf{w}}^{j}_i) \Bigg) \subseteq \text{conv} \Big(  \mathcal{R}_k(x_0^j)  \Big).
\end{equation}
Finally, we define the approximated safe set
\begin{equation*}
   \tilde{\mathcal{SS}}^j = \bigg\{ \bigcup_{k=0}^{T^j} \tilde{\mathcal{R}}_k(x_0^j) \bigg\} \bigcup \mathcal{O},
\end{equation*}
which is used to construct the approximated convex safe set,
\begin{equation}\label{eq:CSapprox}
    \tilde{\mathcal{CS}}^j = \text{conv}\Bigg( \bigcup_{k=0}^j \tilde{\mathcal{SS}}^k \Bigg).    
\end{equation}
It is important to underline that the above approximated convex safe set $\tilde{\mathcal{CS}}^j$ is not invariant, as the approximated reachable sets are an inner approximation of the exact reachable sets (Figure~\ref{fig:safeSetComparisonData}). Indeed, it may exist a disturbance realization which can steer the closed-loop system \eqref{eq:sys} and \eqref{eq:MPCpolicy} outside $\tilde{\mathcal{CS}}^j$. In particular, given $x \in \tilde{\mathcal{CS}}^j$ there is a probability $\epsilon >0$ that the closed-loop system evolves outside $\tilde{\mathcal{CS}}^j$,
\begin{equation}\label{eq:probNotInvariant}
    \text{Pr} (Ax+B\pi^j(x)+w \notin\tilde{\mathcal{CS}}^j | x \in \tilde{\mathcal{CS}}^j) \geq \epsilon.
\end{equation}
In the result section, we show that the above probability is a function of the number of roll-outs used to construct $\tilde{\mathcal{CS}}^j$. In particular as more roll-outs are collected, $\tilde{\mathcal{CS}}^j$ from \eqref{eq:CSapprox} better approximates the convex safe set ${\mathcal{CS}}^j$ from \eqref{eq:CS}.

\begin{figure}[h!]
\centering
\includegraphics[width= \columnwidth]{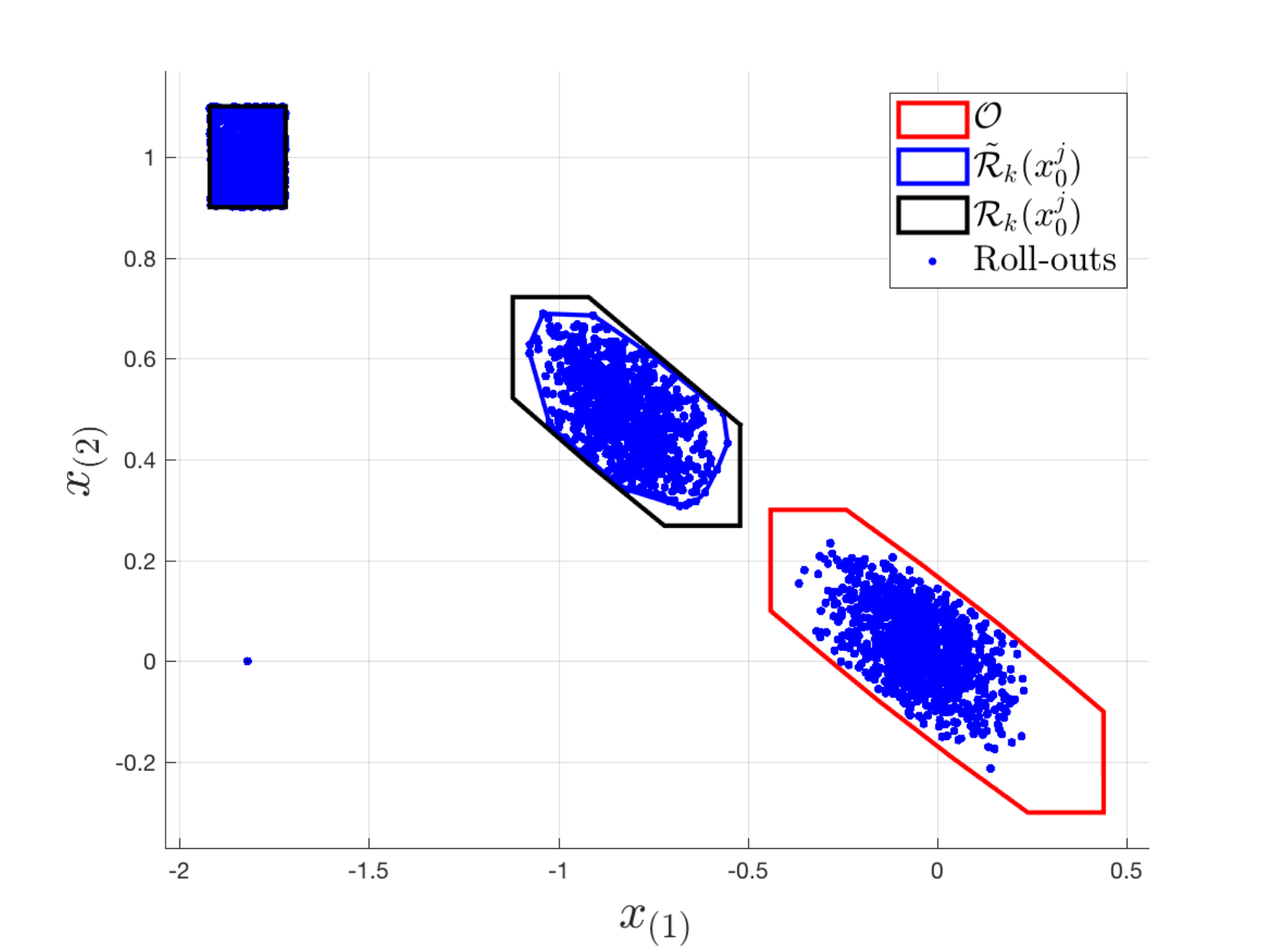}
\caption{Approximated robust reachable sets $\tilde{\mathcal{R}}_k $ from \eqref{eq:reachSetApprox} construct using $1000$ roll-outs. We notice that the approximated robust reachable sets $\tilde{\mathcal{R}}_k $ are an inner approximation the robust reachable sets ${\mathcal{R}}_k $ from~\eqref{eq:reachSet}.} \label{fig:safeSetComparisonData}
\end{figure}

\subsection{Sample-Based Q-function}\label{sec:sampleQfun}
In this section we show how the closed-loop trajectories may be used to approximate the cost-to-go function $L^j_{\pi^j}(\cdot)$ in \eqref{eq:worstCostToGo}. First, we define the realized cost-to-go associated with the stored state $x_k^{j}({\bf{w}}^{i}) \in \tilde{\mathcal{R}}_k(x_0^j) \subseteq \tilde{\mathcal{SS}}^j$,
\begin{equation}\label{eq:realizedCost}
\tilde J^j_{k \rightarrow T^j}(x_k^{j}({\bf{w}}^{i})) = \sum_{t=k}^{T^j} h\Big(x_k^{j}({\bf{w}}^{i}), u_k^{j}({\bf{w}}^{i})\Big)
\end{equation}
where $u_k^{j}({\bf{w}}^{i}) = \pi^j(x_k^{j}({\bf{w}}^{i}))$.

The realized cost \eqref{eq:realizedCost}, associated with the realized trajectory \eqref{eq:realizedTraj}, is used to approximate the worst-case cost-to-go function $L^j_{\pi^j}(\cdot)$. We compute an hyperplane which upper-bounds the realized cost $\tilde J^j_{k \rightarrow T^j}(x_k^{j}({\bf{w}}^{i}))$ for all stored states $\big\{ \bigcup_{i=1}^R x_k^{j}({\bf{w}}^{i}) \big\} \in \tilde{\mathcal{R}}_k(x_0^j)$. In particular, for time $k$ of the $j$th iteration we define the hyperplane $a^j_kx + b^j_k$, where
\begin{equation}
\begin{aligned}
    [a^j_k,b^j_k&] = \\
    =&\underset{a \in \mathbb{R}^n,b \in \mathbb{R}}{\operatorname{argmin}}  \quad \sum_{i=0}^R ||a x_k^{j}({\bf{w}}^{i}) +b -  \tilde J^j_{k \rightarrow T^j}(x_k^{j}({\bf{w}}^{i}))||_2^2\\
    &\quad \quad ~\text{s.t. }  \quad a x_k^{j}({\bf{w}}^{i}) +b \geq \tilde J^j_{k \rightarrow T^j}(x_k^{j}({\bf{w}}^{i})), \\
    &\quad \quad ~\phantom{s.t. } \quad\forall i \in \{0,\ldots,R\}.
\end{aligned}
\end{equation}
At the $j$th iteration, the hyperplanes $a^j_kx + b^j_k$ are used to approximate the worst-case cost-to-go $L^j_{\pi^j}(\cdot)$ from \eqref{eq:worstCostToGo} as follows,
\begin{equation}\label{eq:approxCostToGo}
    \tilde{L}_{\pi^j}^j(x) = \begin{cases}  +\infty & \mbox{If } x\notin \tilde{\mathcal{SS}}^j \\
    0 & \mbox{Elseif } x \in \mathcal{O} \\
    a^j_k x + b^j_k & \mbox{Elseif } x \in   \tilde{\mathcal{R}}_k(x_0^j)
    \end{cases}.
\end{equation}
The resulting approximated value function is defined as 
\begin{equation}\label{eq:QfunApprox}
    \tilde{Q}^j(x) = \min_\mu \mu ~ | ~ (x, \mu) \in  \text{conv}\big(\bigcup_{k=0}^j \text{epi}(\tilde{L}_{\pi^j}(x)^j ) \big).
\end{equation}
Finally, we underline that the above approximated value function is not a control Lyapunov function for system \eqref{eq:sys}. Indeed, there is a probability $\gamma >0$ that Equation~\eqref{eq:proofContrLyap} does not hold and $\tilde{Q}^j(\cdot)$ is not decreasing along the closed-loop trajectory,
\begin{equation}\label{eq:probNotLyap}
    \text{Pr}\big(\tilde{Q}^j(Ax+B\pi^j(x)+w)+h(x,\pi^j(x))-\tilde{Q}^j(x)>0\big)\geq\gamma.
\end{equation}
In the result section, we show that above probability is inversely proportional to the number $R$ of roll-outs used to construct $\tilde{L}_{\pi^j}^j(\cdot)$ in~\eqref{eq:approxCostToGo}.

\section{Results}\label{sec:res}
We test the proposed control strategy on the following double integrator system
\begin{equation}\label{eq:doubleIntSys}
    x_{k+1} = \begin{bmatrix} 1 & 1 \\ 0 & 1 \end{bmatrix} x_k + \begin{bmatrix} 0 \\ 1 \end{bmatrix} u_k + w_k,
\end{equation}
where the the random disturbance $w_k$ is uniformly distributed on the set $\mathcal{W} = \{w \in \mathbb{R}^2: ||w_k||_\infty \leq 0.1\}$. The system is subjected to the following state and input constraints, $x_k \in \mathcal{X} = \{x\in \mathbb{R}^2: ||x||_\infty \leq 10 \}$ and $u_k \in \mathcal{U} = \{u\in \mathbb{R}^2: ||u||_\infty \leq 1 \}$, for all $k\geq0$. Furthermore, we compute the minimal robust positive invariant set $\mathcal{O}$ for the autonomous system $x_{k+1} = (A+BK)x_k + w_k$ where $-K$ is the LQR gain for $Q=1$ and $R=1$. Finally, we define the stage cost $h(x,u) = |x|_\mathcal{O} + |u|_{KO}$
which satisfies Assumption~\ref{ass:cost}.

The convex safe set $\mathcal{CS}^j$ and value function $Q^j(\cdot)$, used in the LMPC~\eqref{eq:FTOCP} and \eqref{eq:MPCpolicy}, are approximated as described in Section~\ref{sec:sampleBased}. In particular at each iteration $j$, we use $R$ roll-outs to compute the approximated safe set $\tilde{\mathcal{CS}}^j$ and value function $\tilde{Q}^j(\cdot)$. In order to initialize the LMPC we set $N=3$, $\tilde{\mathcal{CS}}^0 = \mathcal{O}$ and $\tilde{Q}^0(\cdot) = 0$. Finally at each $j$th iteration, the initial state $x_0^j$ is computed as the furthest point along the negative $x$-axis which belongs to $\mathcal{F}^j$. Basically, we set $a = [-1,~ 0]$ in \eqref{eq:initCondition}.

\subsection{Convex Safe Set and Value Function Approximation}

In this section, we construct $\tilde{\mathcal{CS}}^1$ and $\tilde{Q}^1(\cdot)$ using $R=100$ and $R=1000$ roll-outs. Furthermore, we perform $1000$ Monte-Carlo simulations for the closed-loop system \eqref{eq:sys} and \eqref{eq:MPCpolicy}, in order to estimate the properties of $\tilde{\mathcal{CS}}^1$ and $\tilde{Q}^1(\cdot)$.

\begin{figure}[h!]
\centering
\includegraphics[width= \columnwidth]{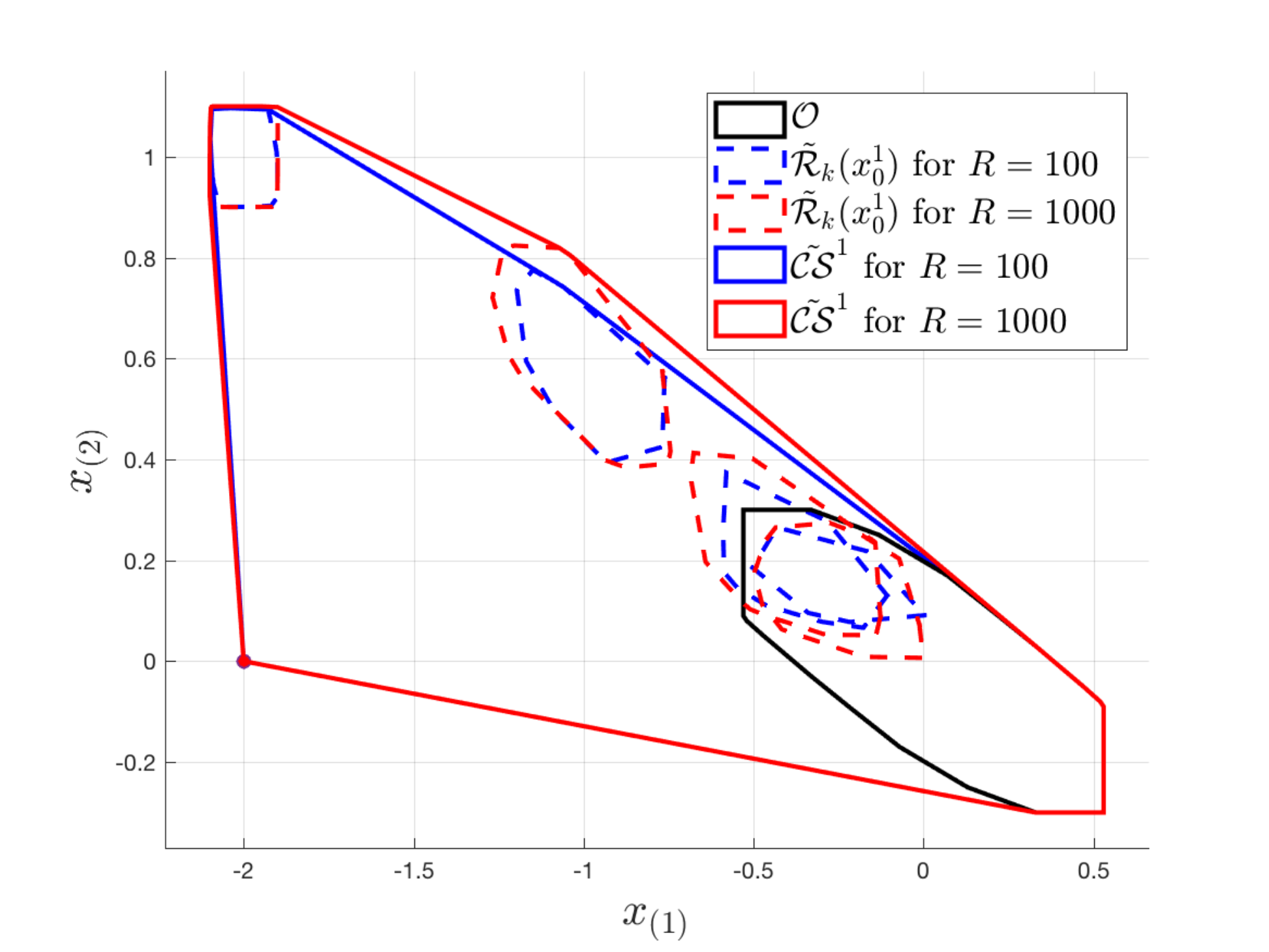}
\caption{The approximated robust reachable sets $\tilde{\mathcal{R}}_k$ \eqref{eq:reachSetApprox} used to construct $\tilde{\mathcal{CS}}^1$ with $R=100$ and $R=1000$ roll-outs. Notice that the approximated convex safe set $\tilde{\mathcal{CS}}^1$ constructed using $1000$ roll-outs contains the one constructed using 100.} \label{fig:safeSetComparison}
\end{figure}

Figure~\ref{fig:safeSetComparison} shows the terminal set $\mathcal{O}$ and the approximated robust reachable sets $\tilde{\mathcal{R}}_k(x_0^1)$, which are used to construct the approximated convex safe set $\tilde{\mathcal{CS}}^j$ with $R=100$ and $R=1000$ roll-outs. As expected, the approximated convex safe set $\tilde{\mathcal{CS}}^j$ constructed using $1000$ trajectories contains the one constructed using $100$ trajectories. As mentioned in Section~\ref{sec:sampleBasedSafeSet} (Eq.~\eqref{eq:probNotInvariant}), the approximated convex safe set is not invariant. Indeed, there is a probability $\epsilon>0$ that, given a state $x \in \tilde{\mathcal{CS}}^1$, the closed-loop system evolves outside $\tilde{\mathcal{CS}}^1$. In order to estimate the probability $\epsilon$, we perform $1000$ Monte-Carlo simulations for the closed-loop system \eqref{eq:sys} and \eqref{eq:MPCpolicy} and we compute the percentage of realized states which evolved outside $\tilde{\mathcal{CS}}^j$. As expected the probability $\epsilon$ decreases as more roll-outs are used to construct $\tilde{\mathcal{CS}}^1$. In particular, we have that $\epsilon \sim 3.6\%$ and $\epsilon \sim 0.3\%$ for $R=100$ and $R=1000$, respectively.

\begin{figure}[h!]
\centering
\includegraphics[trim = 10mm 2mm 10mm 0mm, clip, width=1.0\columnwidth]{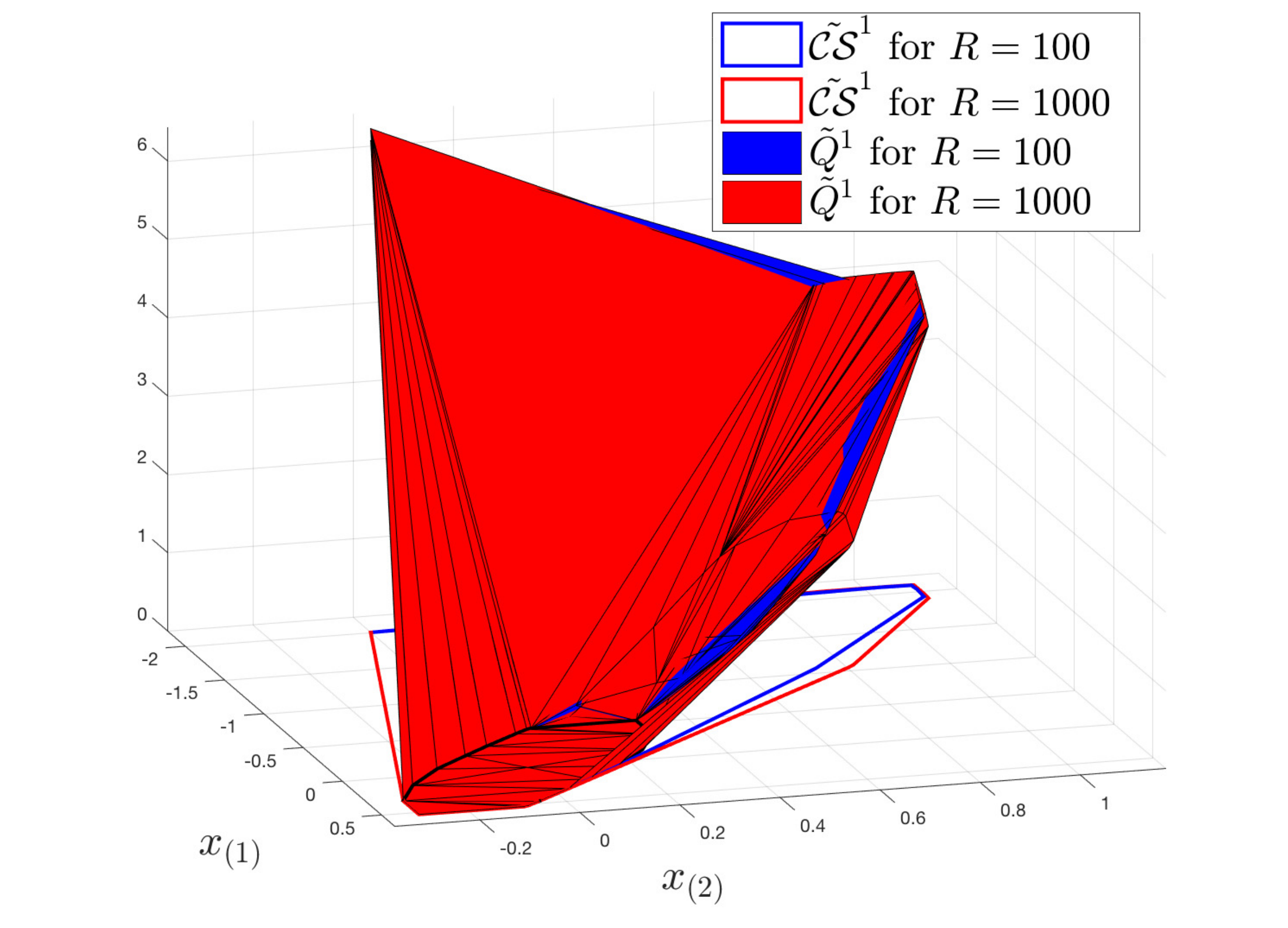}
\caption{Approximated value function $\tilde{Q}^j(\cdot)$ constructed with $R=100$ and $R=1000$ realized trajectories. Note that as more trajectories are used the value of $\tilde{Q}^j(\cdot)$ increases almost everywhere, thus it better approximated ${Q}^j(\cdot)$.}\label{fig:valueFunctionComparison}
\end{figure}

Finally, we analyze how the number of roll-outs affects the approximated value function $\tilde{Q}^1(\cdot)$. Figure~\ref{fig:valueFunctionComparison} shows the approximated value function $\tilde{Q}^1(\cdot)$ constructed with $R=100$ and $R=1000$ roll-outs. First, we notice that the domain of  approximated value function $\tilde{Q}^1(\cdot)$ is enlarged as more realized trajectories are used to compute the approximation. Indeed, the domain of $\tilde{Q}^1(\cdot)$ is the approximated safe set $\tilde{\mathcal{CS}}^1$ from Figure~\ref{fig:safeSetComparison}. Second, we recall that $\tilde{Q}^1(\cdot)$ is constructed based on sampled disturbance sequences and it underestimates $Q^1(\cdot)$, which considers the whole disturbance support. Therefore, we expect that as more sample disturbance sequences are considered $\tilde{Q}^1(\cdot)$ better approximates ${Q}^1(\cdot)$. This intuition is confirmed by Figure~\ref{fig:valueFunctionComparison}, we notice that $\tilde{Q}^1(\cdot)$ constructed with $1000$ trajectories  upper-bounds almost everywhere the value function $\tilde{Q}^1(\cdot)$ constructed with $100$ trajectories, therefore it better approximates $Q^1(\cdot)$. Finally, we recall from Equation~\eqref{eq:probNotLyap} that $\tilde{Q}^1(\cdot)$ is not a robust control Lyapunov function. Indeed, there is a probability $\gamma>0$ that $\tilde{Q}^1(\cdot)$ is not decreasing along the realized closed-loop trajectory. In order to estimate the probability $\gamma$, we use $1000$ Monte Carlo simulations. As expected, the probability $\gamma$ decreases as more closed-loop trajectories are used to construct $\tilde{{Q}}^1(\cdot)$. In particular, we have $\gamma \sim 10.1\%$ and $\gamma \sim 4.3\%$ for $R=100$ and $R=1000$, respectively.

\subsection{Iterative Policy Update}\label{sec:iterativePolicyUpdate}
In this section we run the LMPC for $10$ iterations. In particular, at each $j$th iteration we collect $R=1000$ roll-outs which are used to compute the approximated convex safe set $\tilde{\mathcal{CS}}^j$ and the approximated value function $\tilde{Q}^j(\cdot)$. We show that the LMPC is able to explore the state space while safely steering the system to the terminal set $\mathcal{O}$.

\begin{table}[h!]
\caption{Initial condition $x_0^j$ at each $j\text{th}$ iteration.}
\label{table:IC}
\centering
\begin{tabular}{l|l} \toprule
  \rule{0pt}{12pt} $x^1_0 = -\begin{bmatrix} 2.00 & 0\end{bmatrix}^\top$        & $x^6_0 = -\begin{bmatrix} 9.90 & 0\end{bmatrix}^\top$\\
  \rule{0pt}{12pt} $x^2_0 = -\begin{bmatrix} 5.46 & 0\end{bmatrix}^\top$        & $x^7_0 = -\begin{bmatrix} 9.90 & 0\end{bmatrix}^\top$\\
  \rule{0pt}{12pt} $x^3_0 = -\begin{bmatrix} 6.86 & 0\end{bmatrix}^\top$        & $x^8_0 = -\begin{bmatrix} 9.90 & 0\end{bmatrix}^\top$\\
  \rule{0pt}{12pt} $x^4_0 = -\begin{bmatrix} 9.35 & 0\end{bmatrix}^\top$        & $x^9_0 = -\begin{bmatrix} 9.90 & 0\end{bmatrix}^\top$ \\
  \rule{0pt}{12pt} $x^5_0 = -\begin{bmatrix} 9.90 & 0\end{bmatrix}^\top$        & $x^{10}_0 = -\begin{bmatrix} 9.90 & 0\end{bmatrix}^\top$ \\ \midrule
 \end{tabular}
\end{table}

\begin{figure}[h!]
\centering \includegraphics[trim = 5mm 2mm 5mm 8mm, clip, width=1.0\columnwidth]{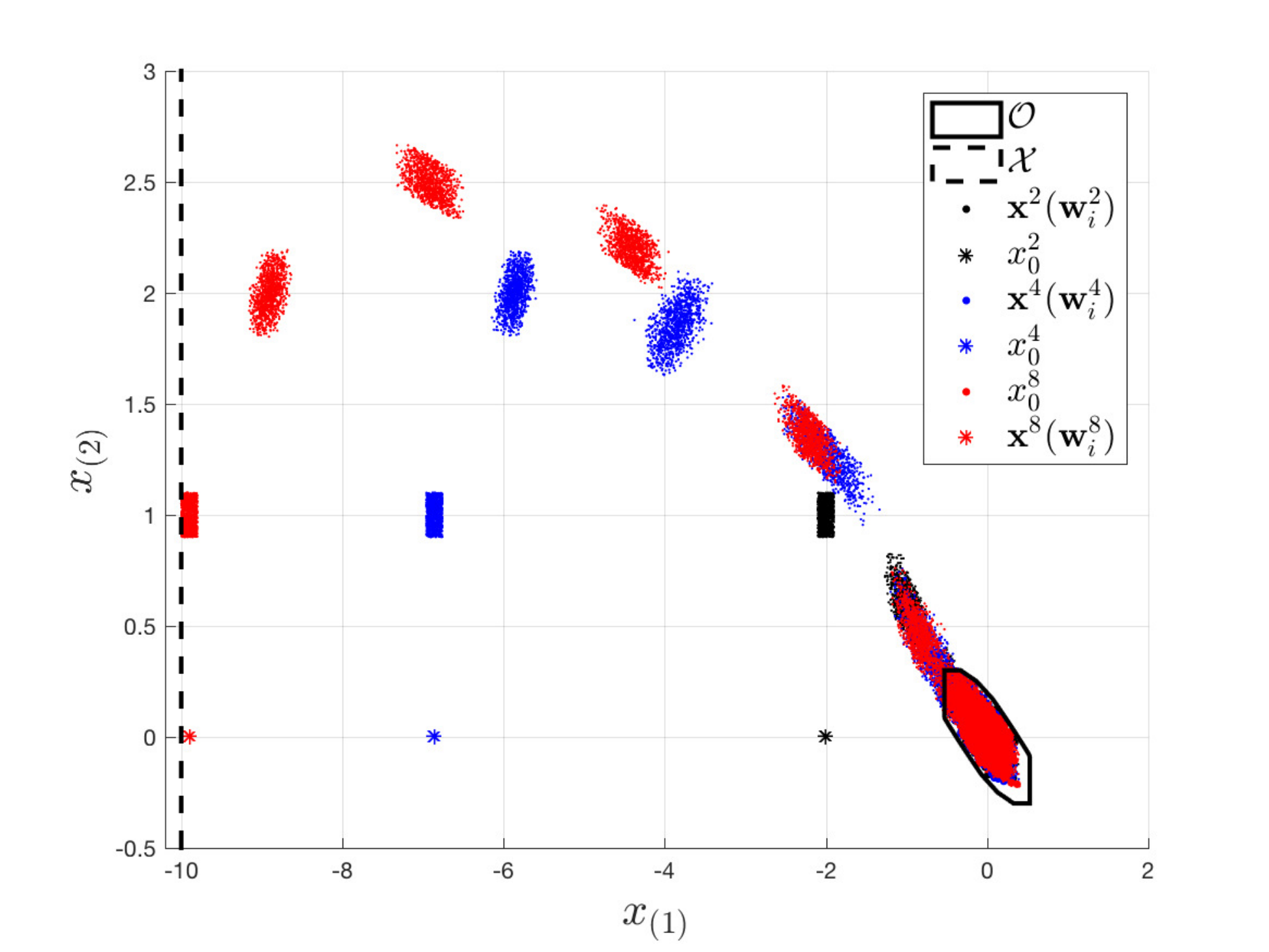}
\caption{For iterations $j\in \{2,4,8\}$ and $i=\{1 ,\ldots,1000\}$ disturbance realizations we show the closed-loop trajectories ${\bf{x}}^j({\bf{w}}^j_i)$ from \eqref{eq:realizedTraj}. Furthermore, we report the initial condition $x_0^j$ which is further from the origin at each iteration.}\label{fig:closedLoopTrajectories}
\end{figure} 

As stated in Section~\ref{sec:res}, at each $j$th iteration we compute the initial condition $x_0^j$ as the furthest point along the negative $x$-axis such that Problem~\eqref{eq:FTOCP} is feasible. Notice that by Theorem~\ref{th:policyDomanin}, the domain of the LMPC policy $\mathcal{F}^j$ is enlarged at each iteration (i.e. ${\mathcal{F}}^k \subseteq {\mathcal{F}}^j$ for all $k \in \{1, \ldots, j\}$). As a result, the region of the state space from which the controller is able to safely complete the control task grows at each iteration. This fact is highlighted in Table~\ref{table:IC}, where we report the initial condition $x_0^j$ as a function of the iteration index. Furthermore, in Figure~\ref{fig:closedLoopTrajectories} we show $1000$ realized trajectories for the $2$nd, $4$th and $8$th iterations. We notice that at each iteration the LMPC safely operates the system over progressively larger regions of the state space, until the closed-loop trajectory is close to saturate the state constraints.

\begin{figure}[b!]
\centering \includegraphics[trim = 5mm 2mm 5mm 8mm, clip, width=1.0\columnwidth]{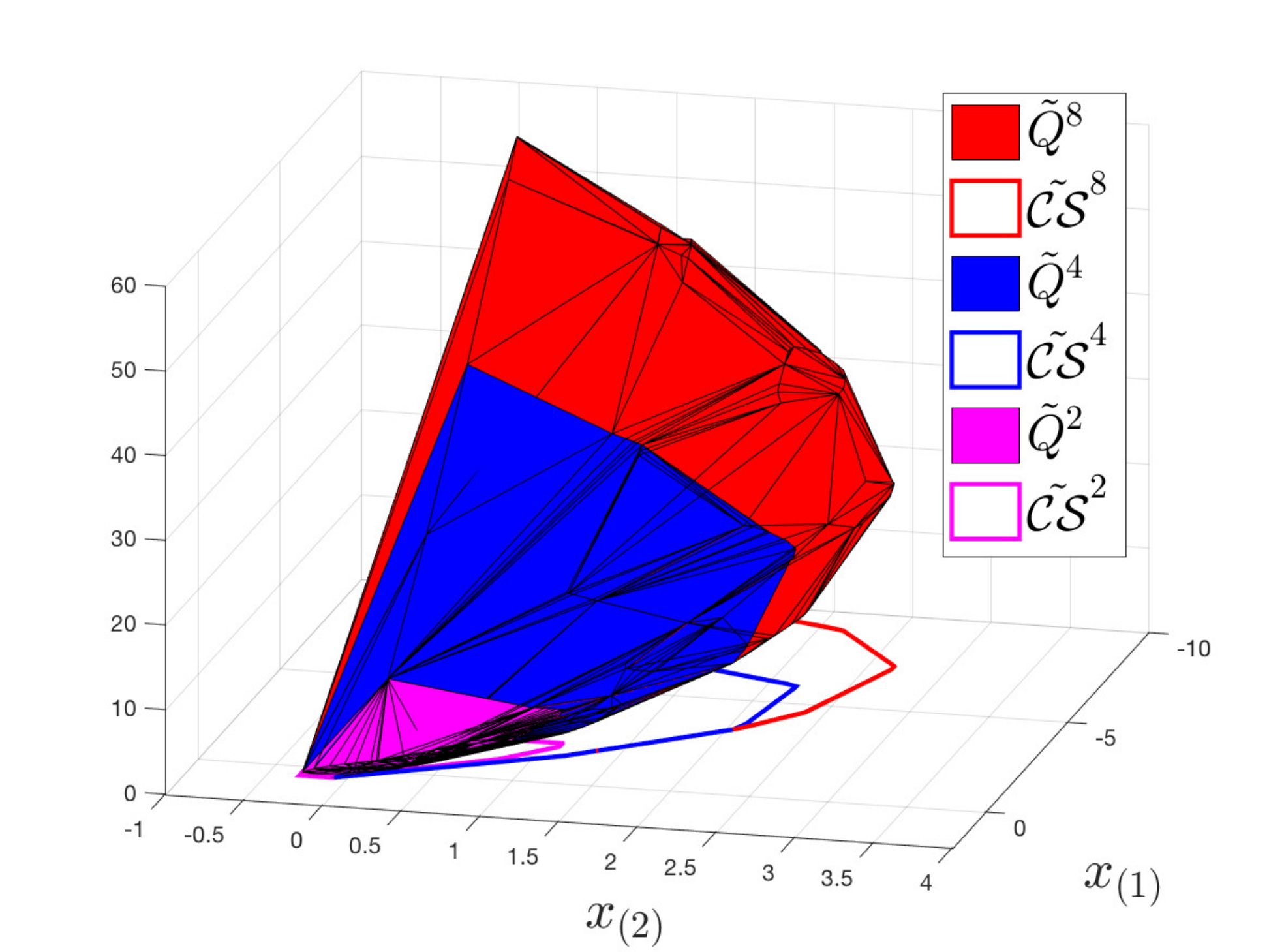}
\caption{Approximated value function $\tilde{Q}^j$ at the $2$nd, $4$th and $8$th iteration. Notice that the domain of $\tilde{Q}^j$ is enlarged at each iteration.}\label{fig:valueFunctionEvolution}
\end{figure}

Finally, in Figure~\ref{fig:valueFunctionEvolution} we report the approximated value function $\tilde{Q}^j(\cdot)$ for the $2$nd, $4$th and $8$th iterations. We recall that the domain of $\tilde{Q}^j(\cdot)$ is the approximated convex safe set $\tilde{\mathcal{CS}}^j$, which is enlarged at each iteration. Therefore, as more iterations of the control task are executed, $\tilde{Q}^j(\cdot)$ approximates the value function over larger regions of the state space, as shown in Figure~\ref{fig:valueFunctionEvolution}.

\subsection{Performance Improvement}

In this section we empirically validate Theorem~\ref{th:cost}. We design a LMPC which minimizes the stage cost $\bar h(x,u) = 0.1|x|_\mathcal{O} + |u|_{KO}$. Afterwards, we run the closed-loop system for $10$ iterations starting from the same initial condition, $x_0^j = -[0,~9.9]~ \forall j \in \{0,\dots,9\}$. In order to initialize the LMPC, we use a suboptimal controller which robustly steers system \eqref{eq:doubleIntSys} to $\mathcal{O}$ and we exploit the closed-loop data to initialize the approximated convex safe set and value function. 

\begin{figure}[t!]
\centering \includegraphics[trim = 5mm 2mm 5mm 8mm, clip, width=1.0\columnwidth]{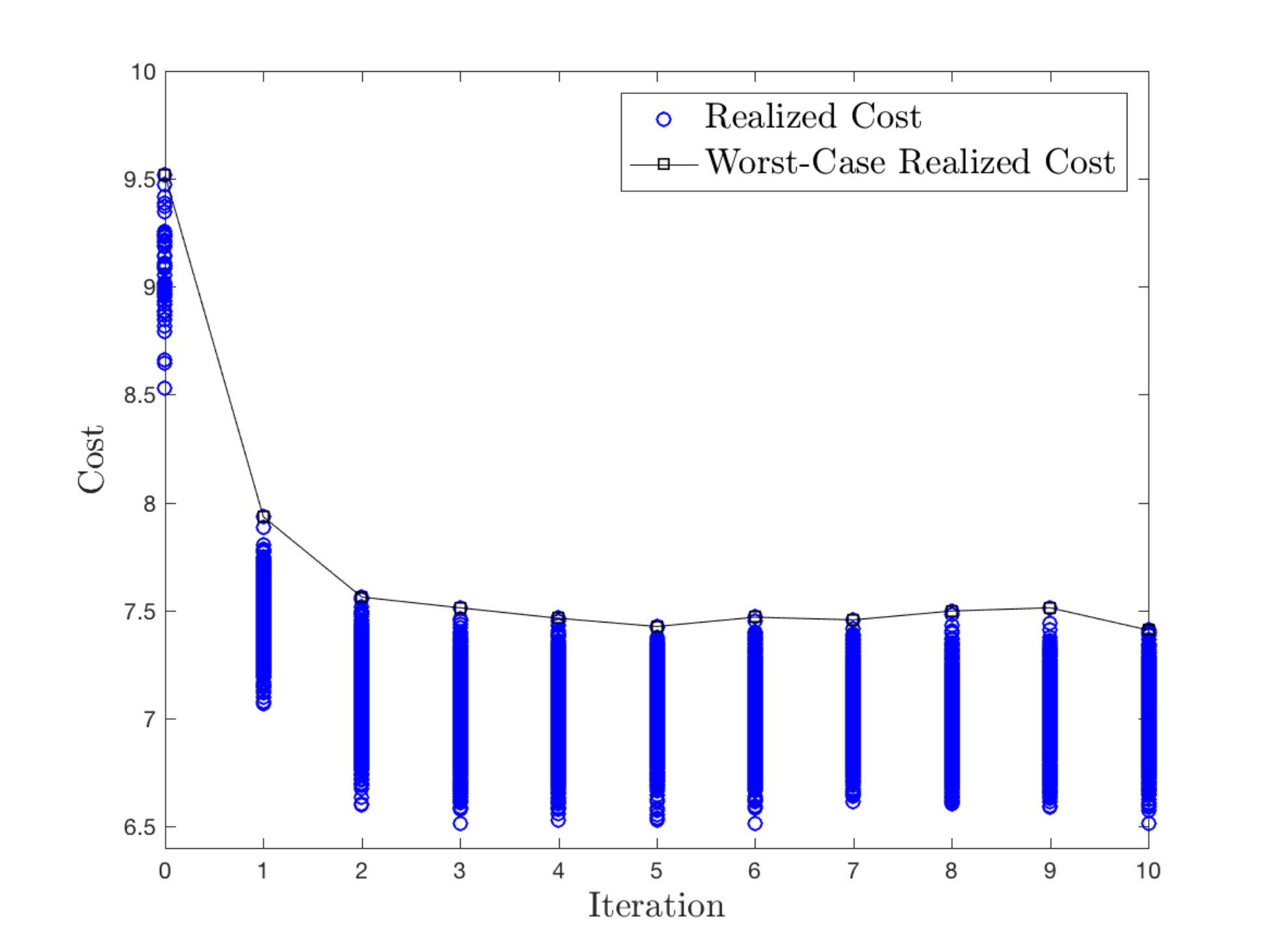}
\caption{Worst-case realized cost and realized cost of the LMPC over the iteration index. We notice that the LMPC improves the worst-case realized cost from the suboptimal controller at the $0$th iteration, until it reaches convergence.}\label{fig:worstCaseCostBound}
\end{figure}

Figure~\ref{fig:worstCaseCostBound} shows the closed-loop cost $\tilde J^j_{0 \rightarrow T^j}(x_0^j({\bf{w}}^j_i))$ from~\eqref{eq:realizedCost} and the worst-case realized cost
\begin{equation}\label{eq:worstCaseRealizedCost}
\max_{i \in \{0,\ldots,R\}} \tilde J^j_{1 \rightarrow T^j}(x_0^j({\bf{w}}^j_i))
\end{equation}
for $10$ iterations. We notice that the LMPC is able to improve the worst-case realized cost associated with the suboptimal policy used at the $0$th iteration. Furthermore, we underline that the controller performs exactly the same task at each iteration ($x_0^j = x_0^i, \forall j,i\geq0$) and the worst-case realized cost \eqref{eq:worstCaseRealizedCost} decreases at each iteration, until it converges within a tolerance of $0.7\%$ as stated in Theorem~\ref{th:cost}.

\section{Conclusions}
In this paper we proposed a sample-based Learning Model Predictive Controller (LMPC) for linear system subject to bounded additive uncertainty. First, we used the LMPC policy to construct a safe set and the associated value function. 
Afterwards, we showed that the proposed strategy allows to guarantee safety and worst-case performance improvement. 
Finally, we exploited sampled closed-loop trajectories to approximate the safe set and associated value function. 
We demonstrated the effectiveness of the proposed approach on a numerical example. In particular, we showed that the proposed LMPC is able to safely explore the state space while estimating the value function associated with the control task. Future work concentrates on finding probability bounds, which would allows to characterize the properties of the approximated safe set and approximate value function as a function of the sampled trajectories.


\bibliographystyle{IEEEtran}
\bibliography{IEEEabrv,mybibfile}

\end{document}